\documentclass[12pt]{article}
\usepackage{authblk}
\usepackage{geometry}
\usepackage{amsmath}
\usepackage{amssymb}
\usepackage{amsthm}
\usepackage{mathrsfs}
\usepackage{customcommands}
\usepackage{bm}
\usepackage{caption}
\usepackage{Nettastyle}
\usepackage{dsfont}
\usepackage{comment}
\usepackage[utf8]{inputenc}
\usepackage{calc}
\usepackage{accents}
\usepackage{amsfonts}

\usepackage{braket}
\usepackage[normalem]{ulem}
\usepackage{color}
\newcommand\al{{\alpha}}
\newcommand\ep{\epsilon}
\newcommand\sig{\sigma}
\newcommand\Sig{\Sigma}

\newcommand\sA{{\ensuremath{{\mathcal A}}}}
\newcommand\sB{{\ensuremath{{\mathcal B}}}}
\newcommand\sC{{\ensuremath{{\mathcal C}}}}
\newcommand\sD{{\ensuremath{{\mathcal D}}}}

\newcommand\sH{{\ensuremath{{\mathcal H}}}}

\newcommand\sM{{\ensuremath{{\mathcal M}}}}

\newcommand\sW{{\mathcal W}}

\newcommand\sY{{\mathcal Y}}

\newcommand{\wt}{\widetilde}

\newcommand{\ee}{\end{equation}}

\newcommand{\eea}{\end{eqnarray}}
\newcommand{\bega}{\begin{gather}}
\newcommand{\eega}{\end{gather}}

\newcommand{\bi}{\begin{itemize}}
\newcommand{\ei}{\end{itemize}}
\newcommand{\ben}{\begin{enumerate}}
\newcommand{\een}{\end{enumerate}}
\newcommand{\bca}{\begin{cases}}
\newcommand{\eca}{\end{cases}}
\newcommand{\bln}{\begin{align}}
\newcommand{\eln}{\end{align}}
\newcommand{\bst}{\begin{split}}
\newcommand{\est}{\end{split}}
\def\ie{\begin{equation}\begin{aligned}}
\def\fe{\end{aligned}\end{equation}}
\newcommand{\bma}{\le(\begin{matrix}}
\newcommand{\ema}{\end{matrix}\ri)}
\newcommand{\bfl}{\noindent\begin{flushleft}}
\newcommand{\efl}{\par\end{flushleft}}

\newcommand\p{\ensuremath{\partial}}

\usepackage{ulem}
\usepackage{mathtools}
\usepackage{tensor} 
\usepackage{tikz}
\usetikzlibrary{arrows, positioning, quotes, intersections}
\usepackage{subcaption}

\usepackage{graphicx}  

\usepackage{thmtools, thm-restate}

\title{The Making of von Neumann Algebras from Bulk Focusing}

\author{Netta Engelhardt}
\author{and Hong Liu}

\affiliation{Center for Theoretical Physics -- a Leinweber Institute, Massachusetts Institute of Technology, \\Cambridge, MA 02139, USA}
\emailAdd{engeln@mit.edu}
\emailAdd{hong\_liu@mit.edu}

\abstract{The single-trace, infinite-$N$ algebra of an arbitrary region may or may not be a von Neumann algebra depending on the GNS sector. In this paper we identify the holographic dual of this mechanism as a consequence of the focusing of null geodesics; more precisely, this GNS sector-dependence corresponds to the well-known difference between null congruences fired from the bulk and those fired from the boundary. As part of establishing this property, we give a rigorous formulation and proof of causal wedge reconstruction for those general boundary subregions whose single trace algebras support von Neumann algebras at large-$N$. We discuss a possible finite-$N$ extension and interpretation of our results as an explanation for the Hawking area theorem.}
\preprint{MIT-CTP/5918}

\begin{document}

\maketitle

\section{Introduction}\label{sec:intro}  
The algebraic approach to the AdS/CFT correspondence
has illuminated many aspects of duality, especially as applied to emergence of the dual bulk geometry and its quantum states (see e.g.~\cite{Reh99a,Reh00,Har16, EngWal18, KamPen19, Fau20, LeuLiu21a,LeuLiu21b,LeuLiu22,Wit21,ChaPen22a,ChaPen22b,FauLi22,JenSor23,KudLeu23,EngLiu23}). An essential part of this is the duality between boundary spatial regions and bulk entanglement (and simple) wedges~\cite{BouFre12, CzeKar12, Wal12, DonHar16, CotHay17, HayPen18}.
Recently, this connection was generalized to arbitrary subregions at infinite-$N$ under the moniker ``subregion/subalgebra duality''~\cite{LeuLiu22}; this serves as a framework through which bulk locality, causal structure, and other geometric properties are encoded in boundary data. In the strict large-$N$ limit, details about bulk causality have been explicitly identified with the emergent type III$_{1}$ von Neumann subalgebra of the boundary CFT~\cite{LeuLiu21a,LeuLiu21b,LeuLiu22}. 

Subregion/subalgebra duality equates a bulk spacetime subregion with an emergent type III$_1$ von Neumann subalgebra of the boundary CFT. 
One important tool for identifying the boundary algebra dual of a bulk region is causal wedge reconstruction~\cite{BanDou98,Ben99,BalKraLaw98,HamKab05,HamKab06,KabLif11,Hee12,HeeMar12,PapRaj12,Mor14,EngPen21a, Wit23}: given some boundary domain of dependence $D[\sigma]$, the subalgebra associated to it should be a von Neumann algebra. The corresponding bulk region is proposed to be the causal completion of the causal wedge: $(J^{+}_{\rm bulk}[D[\sigma]]\cap J^{-}_{\rm bulk}[D[\sigma]])''$.  
Causal wedge reconstruction can be justified using the timelike tube theorem~\cite{Bor61,Ara62,StrWit23b,Str00, Wit23} and is expected to be generally implemented via the HKLL procedure~\cite{HamKab05,HamKab06,Mor14}. The latter is an explicit protocol that requires solving a non-standard Cauchy problem using the bulk equations of motion; see~\cite{Wit23} for a discussion. 

A primary motivation for undestanding holography for the causal wedge is its relevance to the Hawking area theorem~\cite{Haw71}, as the emergence of this universal law of gravity remains poorly understood in general, in contrast with the area law for apparent horizons (``holographic screens''~\cite{Bou99b})~\cite{BouEng15a, BouEng15b, EngWal17b, EngWal18}. If causal  wedges and their areas had clear holographic descriptions and the latter admitted generalizations to arbitrary time bands, we would have a fundamental quantum explanation for the Hawking area law. Information-theoretic approaches to this problem have thus far failed~\cite{EngWal17a}, so it is natural to switch an algebraic route instead.

However, while the algebras of domains of dependence of (closed, acausal) regions are expected to be von Neumann,\footnote{So long as any such region is taken to lie on a single Cauchy slice~\cite{HarShaTA}.} the situation is more subtle for more general regions. While the algebra of a causally concave region is unlikely to ever be von Neumann, causally \textit{convex} regions\footnote{These are regions with the property that any causal curve with endpoints in the region lives entirely in the region.} that are not domains of dependence can in fact support von Neumann subalgebras at infinite-$N$ (i.e. for generalized free fields). For example, consider a boundary time band in a CFT state dual to a spherically symmetric bulk. Let $I_w$ be the time band, which is of width $w < \pi R$ -- i.e., the spacetime region with $t \in \left(-\frac{w}{2}, \frac{w}{2}\right)$ -- on the boundary of global AdS, where $R$ is the AdS radius.
 The causal wedge\footnote{Strictly speaking ``causal wedge'' is typically defined for a boundary domain of dependence and is generically not causally complete; for arbitrary $Y$ and to include the causal completion we call $(J^{+}[Y]\cap J^{-}[Y])''$ a \textit{generalized} causal wedge. } of this region, $\sW_{\rho_w} \equiv J^+_{\rm bulk} (I_w) \cap J^-_{\rm bulk} (I_w)$, 
 is a spherical Rindler region with  ``radius'' $\rho_w$ given by $\rho_w =R \tan \left({\frac{\pi}{2}} - \frac{w}{ 2R} \right)$. Note that $\sW_{\rho_w} = \sW_{\rho_w}''$  is a bulk domain of dependence. 
Extending causal wedge reconstruction to this case, we find the identification~\cite{BanBry16,LeuLiu22}: 
\be \label{timeBI}
  \sM_{\sW_{\rho_w}} = \sY_{I_w}  ,
\ee
where $  \sM_{\sW_{\rho_w}}$ is operator algebra of bulk quantum field theory (in the $G_N \to 0$ limit) in $\sW_{\rho_w}$ and $\sY_{I_w}$ is the algebra generated by single-trace operators in $I_w$ (in the large $N$ limit).  

Taking commutant on both sides of~\eqref{timeBI} and assuming the bulk quantum field theory obeys the Haag duality, we find~\cite{LeuLiu22} 
\be 
 \sM_{\sD_{\rho_w}} = \sY_{I_w}' 
\ee
where $\sD_{\rho_w}$ the bulk causal complement of $\sW_{\rho_w}$ -- i.e., the spherical diamond region in the ``middle'' of global AdS.
A diamond region in the bulk that does not intersect the boundary is thus dual to the commutant of a time-band algebra. By considering increasing values of $w < \pi R$, we can use $\sY_{I_w}'$ to describe progressively smaller diamond regions in the bulk. This setup probes locality in AdS and illustrates how the commutant structure of the boundary algebra encodes the bulk causal structure. The discussion can be generalized to algebras of homogeneous time bands in other spherically symmetric states, such as the thermofield double state or states corresponding to black holes formed from spherically symmetric collapse.  
Such time-band algebras and their commutants have also been used to investigate bulk locality and causal structure in the stringy regime~\cite{GesLiu24} (see also~\cite{JenRaj24}).

Consider a causally convex region $Y$ (including e.g. a time band) and the corresponding algebra of single trace operators ${\cal A}_{Y}$ associated with the region. We denote the representation of ${\cal A}_{Y}$ in a GNS-sector as ${\cal Y}_{Y}$. 
In a slight abuse of notation, we say ${\cal Y}_{Y}$ is von Neumann if its double commutant does not involve single-trace operators of a larger region.\footnote{The double commutant will of course introduce new operators to facilitate the closure; our focus here is on identifying the subregions for which the double commutant introduces new single-trace operators and those  subregions for which it does not.} 
Interestingly, whether ${\cal Y}_{Y}$ is von Neumann or not appears to be state-dependent, or more precisely GNS sector dependent~\cite{LeuLiu24}. On a given GNS Hilbert space ${\cal H}$ at large-$N$, ${\cal Y}_{Y}$ may be von Neumann, while on a different GNS sector ${\cal H}'$, ${\cal Y}_{Y}$ may not be von Neumann. For example, in non-spherically symmetric spacetimes, the homogeneous time band algebra may not be identical to its own double commutant. Another example of the union of two causal diamonds at finite and infinite temperatures is illustrated in Fig.~\ref{fig:diamonds}: at any finite temperature, the double commutant of the union of the two causal diamonds $D[A_{1}]\cup D[A_{2}]$ introduces new single-trace operators localized on a larger region. At infinite temperature, however, the double commutant introduces no new geometrical elements.

A necessary condition for the application of subregion/subalgebra duality is that the boundary subalgebra in question be von Neumann. Thus the identification of all regions that support a von Neumann algebra in a given GNS sector is of paramount importance for understanding bulk reconstruction. Since in AdS/CFT different GNS sectors correspond to different bulk geometries at large-$N$, understanding the dual bulk mechanism mirroring this GNS sector-dependence has direct implications for the emergence of the spacetime geometry. In this paper we thus focus on the following two questions: under what circumstances can one rigorously establish that the boundary algebra is equivalent to the bulk operator algebra associated with its causal wedge? What is the most general class of boundary regions for which such reconstruction holds? 

As we will prove, the bulk manifestation of the change in the regions that support von Neumann algebras between different GNS sectors is the pattern of caustics in null congruences. A well-known effect of null geodesic focusing is that a null congruence fired from boundary to bulk will in general not coincide with a null congruence fired from bulk to boundary. More precisely, the generators of $\partial J_{\rm bulk}^{-}[Y]$ are past-directed null geodesics fired from $Y$ up to intersections to the past. The congruence fired from a cross-section of $\partial J_{\rm bulk}^{-}[Y]$ back towards $Y$ is generated by future-directed null geodesics up to intersections to the future. These will not agree on any generators that have intersections, which are generic.~\footnote{This phenomenon is the reason that the causal and entanglement wedges can both end on the same boundary region despite the latter being generically much deeper in the bulk than the former.} For spacetimes not protected by Killing symmetries or even then for regions that do not respect these symmetries, there will be generators that reach intersections. We find that this phenomenon, which results in different mismatched generators in different geometries, is the bulk realization of the sector-dependence discussed above. As a consequence, we provide general holographic criteria under which causal wedge reconstruction is valid, and the associated time band or more generally causally convex region, is a von Neumann algebra.

\begin{figure}
    \centering
    \includegraphics[width=0.7\linewidth]{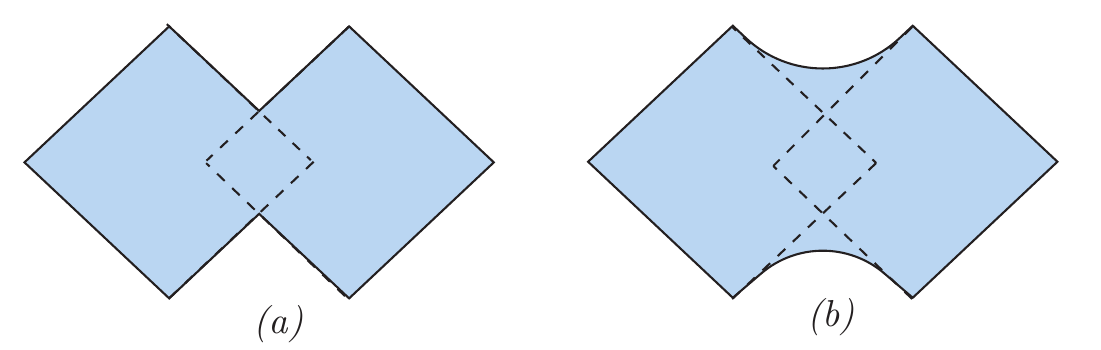}
    \caption{
    (a): the union of two intersecting causal diamonds in $(1+1)$-dimension.
    The single-trace algebra in this region is not a von Neumann algebra in the vacuum state, but is one in the thermofield double state in the infinite temperature limit~\cite{LeuLiu24}. (b): From an old result of Araki~\cite{Ara62},  the algebra of single-trace operators in the region shown in the plot is a von Neumann algebra in the vacuum state. In the thermofield double state of a general temperature, the region which gives a von Neumann algebra lies between (b) and (a).     }
    \label{fig:diamonds}
\end{figure}

We now summarize our result, which is given in a number of theorems in the main text. 
The algebra $\sY_Y$ of single-trace operators associated with a boundary spacetime region $Y$ admits a standard reconstruction of the causal wedge and is a von Neumann algebra if and only if the following condition is true: denoting the causally complete, generalized causal wedge
\be \label{dedsC}
\sC_Y = \left( J_{\rm bulk}^{+}[Y]\cap  J_{\rm bulk}^{-}[Y] \right)'' ,
\ee
and the boundary manifold as $B$, we must have 
\be \label{keycon}
\sC_Y \cap B=Y \ .
\ee
For such a $Y$, we prove that causal wedge reconstruction holds 
\be \label{gcwr}
\sY_Y  =  \sM_{\sC_Y} 
\ .
\ee
Taking the commutant on both sides of~\eqref{gcwr} gives
\be \label{cdiaD}
\sY_Y'  =  \sM_{\sC_Y}' = \sM_{\sC_Y'} ,
\ee
where in the second equality we have again assumed bulk Haag duality. Equation~\eqref{cdiaD} identifies the bulk causal complement  of $\sC_Y$ with the commutant  of $\sY_Y$. 

When $Y$ does not satisfy the conditions above, the corresponding single-trace operator algebra $\sY_Y$ is not a von Neumann algebra. The minimal von Neumann algebra containing $\sY_Y$ is given by its double commutant $\sY_Y''$. We give a general description of the bulk subregion dual to $\sY_Y''$. This can be viewed as a refinement and realization of the conjectures of~\cite{FauLi22} (see also~\cite{LeuLiu22}). 

We also give a more speculative argument that time-band algebras can be defined at finite-$N$, and the area of the edge of the causal wedge gives the entanglement entropy for the corresponding algebra. This paves the way for an algebraic interpretation of a generalized second law at arbitrary small time steps in fully dynamical spacetimes in which the area of the event horizon changes at leading order as a function of time. An algebraic derivation of the Generalized Second Law was obtained for horizons with no leading order change to the area in~\cite{Wal09,Wal10, Wal11, FauSpe24}, but this still falls short of an algebraic derivation of the leading order component to the GSL in broad generality. 

The plan of the paper is as follows. Section 1.1 outlines the assumptions and conventions that will be used throughout. Section 2 includes the main technical theorems about bulk causal structure. In Section 3 we prove causal wedge reconstruction for $Y=D[C_{Y}]\cap B$, and Section 4 gives a heuristic argument for the finite-$N$ case. A number of technical causal structure lemmas are proved in Appendix A.     

\subsection{Assumptions and Conventions}

We begin with our assumptions. Let $(M,g)$ be an asymptotically AdS spacetime and let $B$ be its asymptotic (AdS) boundary. We make the following assumptions about $M$ and $B$:
\begin{itemize}
    \item $M\cup B$ is globally hyperbolic and $B$ is separately globally hyperbolic.
    \item The causal structure of $M$ respects the causal structure of $B$: i.e boundary-to-boundary causal curves cannot travel faster through $M$ than on $B$.\footnote{This is guaranteed by the Null Convergence Condition (NCC)~\cite{GaoWal00} but merely requiring boundary causality of the bulk is in fact a weaker constraint than the NCC~\cite{EngFis16}; in the interest of making as few assumptions as possible, we assume only boundary causality.}
    \item All Cauchy slices of both $M\cup B$ and $B$ will always be taken to be acausal.
    \item The AdS boundary $B$ is taken to be spatially compact.
    \item We work in the infinite-$N$ and $\lambda$ limit of AdS/CFT: the bulk is described by quantum fields propagating on a fixed curved background.
\end{itemize}
We now state our conventions:
\begin{itemize}
    \item If $Q\subset M\cup B$, $\partial Q$ will denote the boundary of $Q$ in $M\cup B$, that is:
    $$\partial Q = \partial Q|_{M}\cup (Q\cap B).$$
    We will sometimes abuse notation: for an achronal set $Q$ which is compact in $M\cup B$, we will sometimes use $\partial Q$ to denote the boundary of $Q$ within any Cauchy slice that contains it.  
    \item Given a set $Q\subset M\cup B$ or  $Q\subset B$, we will use $J^{\pm}[Q]$ to denote the causal future/past of $Q$ in $M\cup B$ (regardless of whether $Q\subset M\cup B$ or $Q\subset B$). If $Q\subset B$, we will use $J_{B}^{\pm}[Q]$ to denote the causal future/past of $Q$ in $B$.
    \item To denote the boundary of a set $X\subset B$, we use $\partial_{B}X$.
    \item We define Edge$[Q]=\partial_{B}(Q\cap B)$. This differs from the topological definition in e.g.~\cite{Wald}, which we will refer to as lowercase ``edge'' rather than ``Edge''. 
    \item A spacetime region $Q$ (in either $B$ or $M\cup B$) will be said to be causally convex if any causal curve with endpoints in $Q$ lies entirely in $Q$. Equivalently, $Q$ is causally convex if it is globally hyperbolic as a separate spacetime manifold~\cite{BosFew20}: i.e. it is given by a Cauchy development of a surface $\sigma$ that is Cauchy-splitting, where said Cauchy development \textit{may be non-maximally extended}. That is, there exists an acausal Cauchy-splitting surface $\sigma$ that splits $Q$ into $J^{\pm}[\sigma]\cap Q$ where every inextendible causal curve that intersects $Q$ crosses $\sigma$. See Fig.~\ref{fig:causallyconvex} for examples of such regions. We denote domains of dependence of regions $X\subset Q$ within $Q$ as $D_{Q}[X]$. Thus $Q$ is causally convex if $D_{Q}[\sigma]=Q$.  We emphasize $Q$ can be causally convex even when it is not a complete domain of dependence within $M\cup B$ or within $B$. 

    \begin{figure}
        \centering
        \includegraphics[width=\linewidth]{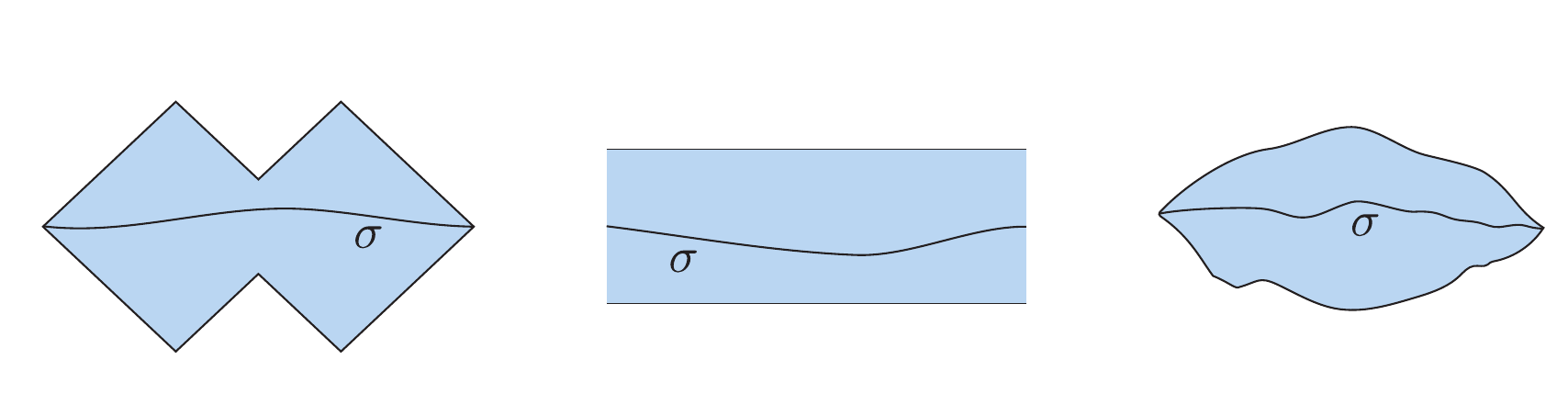}
    \caption{Three different causally convex regions: on the left, the union of two causal diamonds; in the middle, a time band; on the right: a wiggly region. A ``Cauchy surface'' $\sigma$ is marked for each region.}
        \label{fig:causallyconvex}
    \end{figure}
    \item Let $Q$ be causally convex with a Cauchy slice $\sigma$. Denote $Q^{\pm}=D_{Q}^{\pm}[\sigma]$ and $\partial^{\pm}[Q]=\partial {\cal H}^{\pm}_{Q}$ where ${\cal H}^{\pm}_{Q}$ is the future/past Cauchy horizon of $Q$. Note that the division into $Q^{\pm}$ is Cauchy-slice-dependent. 
    \item All domains of dependence are taken to be closed.
    \item The complete causal cone of a point $p$ or a set $S$ is denoted $J(p)$ or $J[S]$, respectively.
    \item The algebra generated by single-trace operators in a boundary spacetime region $Y$ in the infinite-$N$ limit  is denoted ${\cal Y}_{Y}$. The \textit{bulk} operator algebra in a spacetime region $U$ is denoted ${\cal M}_{U}$. 
    \item All other notation and conventions are as in~\cite{Wald}.
\end{itemize}

\section{Identifying Bulk and Boundary Causal Regions}
Below we first state (Sec. 2.1) and then prove (Sec. 2.2) the causal structure theorems necessary for our main result. 

\subsection{Main Technical Theorems}
\begin{thm}\label{Thm:Forwards} Let $Y\subset B$ be  causally convex and compact. Assume that there is no Cauchy slice of $M\cup B$ that is fully contained in both $J^{+}[Y]$ and $J^{-}[Y]$. Then there exists a bulk Cauchy slice $\Sigma$ and a Cauchy slice $\sigma$ of $Y$ such that $C_{Y} =J^{+}[Y]\cap J^{-}[Y]\cap \Sigma$ satisfies 
\ben
 \item $D[C_{Y}]=C_{Y}''$;
\item $\partial C_{Y}=(\partial J^{+}[Y]\cap \partial J^{-}[Y])\cup \sigma$.
\een 
Furthermore, for any such $C_{Y}$, $Y\subseteq D[C_{Y}]\cap B$.\footnote{A heuristic analysis pointing to this conclusion for the special case of a boundary time strip was discussed in~\cite{Hub14} in the context of hole-ography~\cite{BalCze13a, BalCho13b}.}
\end{thm}

\begin{thm} \label{thm:Cexists}
Let $C$ be any acausal closed bulk hypersurface satisfying 
\be \label{eq:Ceqn}
\partial C = (\partial J^{+}[Y]\cap \partial J^{-}[Y])\cup \sigma  \ ,
\ee
where $\partial C\cap M\neq \varnothing$ and $Y\subset B$ is causally convex.. Then there exists a Cauchy slice $\Sigma$ of $M\cup B$ such that 
\be
C_{Y}=J^{+}[Y]\cap J^{-}[Y] \cap \Sig = J^{+}[Y]\cap \Sigma=J^{-}[Y]\cap\Sigma  \ ,
\ee
where $\partial C_{Y}=\partial C$ and $D[C]=D[C_{Y}]$.
\end{thm}

\begin{thm}\label{Thm:Forwards2}
Let $Y$ satisfy the assumptions of Theorem~\ref{Thm:Forwards} and let $\Sigma$ be the bulk Cauchy slice whose existence is guaranteed by Theorem~\ref{Thm:Forwards}, i.e.
$$C_{Y}=J^{+}[Y]\cap J^{-}[Y]\cap\Sigma$$
satisfies (1) and (2) of Theorem~\ref{Thm:Forwards}. Define $Y_{\rm max}\equiv D[C_{Y}]\cap B$. Then there exists a choice of Cauchy slice $\widetilde{\Sigma}$ such that $$C_{Y_{\rm max}}\equiv J^{+}[Y_{\rm max}]\cap J^{-}[Y_{\rm max}]\cap \widetilde{\Sigma}$$ satisfies 
$$\partial C_{Y_{\rm max}}=\partial C_{Y}.$$
In particular, it is possible to pick $\Sigma=\widetilde{\Sigma}$ so that $C_{Y}$ still satifies (1) and (2) above and $C_{Y}=C_{Y_{\rm max}}$.
\end{thm}

\begin{thm}\label{thm:minimalC}
    Let $\wt C$ be any closed bulk acausal hypersurface with boundary satisfying $D[\wt C]\cap B=Y$, where $Y$ is causally convex and compact in $M\cup B$. Then $D[C]\subset D[\wt C]$ where $C$ is any acausal bulk hypersurface satisfying
    $$\partial C=(\partial J^{+}[Y]\cap \partial J^{-}[Y])\cup \sigma$$
    and $\sigma$ is any Cauchy slice of $Y$. Moreover, $D[C]\cap B=Y$. 
\end{thm}

\begin{thm} \label{thm:Envelope}
Let $X_{Y}$ be the timelike envelope\footnote{That is, $X_{Y}$ consists of those points $x\in M\cup C$ that are contained in a causal curve $\gamma\subset M\cup B$ between points $p,q\in Y$ such that $\gamma$ is fixed-endpoint causally homotopic to a curve that lies entirely in $Y$.} of $Y$, with $Y$ as in Theorem~\ref{Thm:Forwards}. Then:
\begin{enumerate}
    \item $X_{Y}\subset D[C_{Y}]$;
    \item $X_{Y}$ contains a Cauchy slice of $D[C_{Y}]$.
\end{enumerate} 
\end{thm}

We now give proofs of the above theorems. All lemmas quoted below are stated and proved in Appendix~\ref{app:lemmas}.

\subsection{Proofs of Theorems 1-5}

\textbf{Proof of Theorem 1.}

\begin{proof} Point (1) follows immediately from Lemma~\ref{CausalCompletionLem}, since $C_{Y}$ is achronal, compact, closed in $M\cup B$, and admits support on a single Cauchy slice. By Lemma~\ref{lem:New013}, there exists a Cauchy slice $\Sigma$ such that $C_{Y}=J^{+}[Y]\cap J^{-}[Y]\cap \Sigma$ satisfies $\partial C_{Y}=(\partial J^{+}[Y]\cap \partial J^{-}[Y])\cup \sigma$. This proves point (2). 

So we know that our $\Sigma$ contains $\partial C_{Y}=\partial J^{+}[Y]\cap \partial J^{-}[Y]=\partial J^{+}[Y^{-}]\cap \partial J^{-}[Y^{+}]\subset C=J^{+}[Y^{-}]\cap J^{+}[Y^{+}]\cap \Sigma$. Let $p\in Y^{+}$ and let $\gamma(\lambda)$ be any past-directed past-inextendible causal curve starting at $p=\gamma(1)$. Since $\gamma(1)\in Y^{+}\subset D_{B}^{+}[\sigma]$, $\gamma$ must eventually cross $\Sigma$, as it is a Cauchy slice that contains $\sigma$. Let $A\equiv \Sigma/C_{Y}$, where $C_{Y}$ is closed. Then $\gamma$ can cross $\Sigma$ either on $A$ or on $C_{Y}$. If every such $\gamma$ crosses $\Sigma$ on $C_{Y}$, then $p\in D^{+}[C_{Y}]$ $\forall p \in Y^{+}$, and thus $Y^{+}\subset D^{+}[J^{+}[Y]\cap J^{-}[Y]\cap \Sigma]\cap B$; similarly for $p\in Y^{-}$, and since we are working with $\Sigma$ where $\sigma \in C_{Y}$, $p\in \sigma$ are automatically in $C$. So if every $\gamma$ crosses $\Sigma$ on $C_{Y}$, then $Y\subset D[C_{Y}]$ and we are done. 

Suppose that $\gamma$ crosses $\Sigma$ at $A$ instead. This immediately implies that $J^{-}[Y^{+}]\cap \Sigma\neq C_{Y}$, in contradiction with Lemma~\ref{Completeness}. So every such $\gamma$ crosses $\Sigma$ at $C_{Y}$. This establishes that $Y\subseteq D[C_{Y}]\cap B.$
\end{proof}

\noindent \textbf{Proof of Theorem 2.}

\begin{proof}
Let $\Sigma$ be any Cauchy slice of $M\cup B$ containing $(\partial J^{+}[Y]\cap \partial J^{-}[Y])\cup \sigma$ (by Lemma~\ref{Lemma7} such a Cauchy slice exists). Since by Lemma~\ref{LemmaNew}, $\partial J^{+}[Y]\cap \partial J^{-}[Y]$ is a complete cross-section of $\p J^{\pm}[Y]$, and $\partial J^{+}[Y]\cap \partial J^{-}[Y]\cap B=\partial_{B}\sigma$ by Lemma~\ref{lemma3}, $\partial J^{+}[Y]\cap \partial J^{-}[Y]$ must be Cauchy splitting (both in $M$ and in $M\cup B$). This follows since $\partial J^{\pm}[Y]$ are causal Cauchy surfaces that are split by $\partial J^{+}[Y]\cap \partial J^{-}[Y]$; in  a globally hyperbolic spacetime the topology of Cauchy slices immediately yields that $\partial J^{+}[Y]\cap \partial J^{-}[Y]$ must also be Cauchy splitting for an acausal Cauchy slice. Define $\wt C$ to be the hypersurface bounded by $\sigma=Y\cap \Sigma$ and $\partial J^{+}[Y]\cap \partial J^{-}[Y]$. This exists because $\partial J^{+}[Y]\cap \partial J^{-}[Y]$ and $\sigma$ are homologous. Then the intersection $J^{+}[Y]\cap J^{-}[Y]\cap D[C]$ contains $\partial C\subset J^{+}[Y]\cap J^{-}[Y]$. Thus (by the Cauchy splitting property) there exists a Cauchy slice $\wt \Sigma$ such that $\wt C=C_{Y}=J^{+}[Y]\cap J^{-}[Y]\cap \wt \Sigma$, $ \wt C=\partial C$, and $D[C]=D[\wt C]$.

\end{proof}

\noindent \textbf{Proof of Theorem 3.}
\begin{proof} We first prove that (1) $J^{+}[Y_{\rm max}]\cap J^{-}[Y_{\rm max}]\subset D[C_{Y}]$, and then we prove that (2) $C_{Y}\subset J^{+}[Y_{\rm max}]\cap J^{-}[Y_{\rm max}]$. This immediately implies that $D[C_{Y}]=D[J^{+}[Y_{\rm max}]\cap J^{-}[Y_{\rm max}]]$. By Lemmas~\ref{Lemma7}~\ref{LemmaNew}, there exists a Cauchy slice $\wt \Sigma$ such that $D[J^{+}[Y_{\rm max}]\cap J^{-}[Y_{\rm max}]]=D[J^{+}[Y_{\rm max}]\cap J^{-}[Y_{\rm max}]\cap \wt \Sigma]=D[C_{Y_{\rm max}}]$. This establishes $D[C_{Y}]=D[C_{Y_{\rm max}}]$, and thus edge$[C_{Y}]={\rm edge}[C_{Y_{\rm max}}] $. In the notation used in this paper, this is precisely the statement that $\partial C_{Y}=\partial C_{Y_{\rm max}}$. That is, it is possible to pick $\Sigma$ and $\wt \Sigma$ such that $\partial C_{Y}=\partial C_{Y_{\rm max}}$. Let $\wt C_{Y}=J^{+}[Y]\cap J^{-}[Y]\cap \wt \Sigma$. Because $\partial C_{Y}\subset \wt \Sigma$, and $\partial C_{Y}\cap M=\partial J^{+}[Y]\cap \partial J^{-}[Y]$, $\partial C_{Y}\subset \wt C_{Y}$. Because $\partial J^{+}[Y]\cap \partial J^{-}[Y]$ is a complete cross-section of the $J^{\pm}[Y]$ congruences by Lemma~\ref{LemmaNew}, $\partial C_{Y}=\partial \wt C_{Y}$. Thus $\wt C_{Y}$ satisfies the properties in Theorem 1.

We now prove (1): that $J^{+}[Y_{\rm max}]\cap J^{-}[Y_{\rm max}]\subset D[C_{Y}]$. By contradiction. Suppose that $J^{+}[Y_{\rm max}]\cap J^{-}[Y_{\rm max}]$ were not in $D[C_{Y}]$. Then there would exist a past- and future-inextendible causal curve $\gamma$ through $J^{+}[Y_{\rm max}]\cap J^{-}[Y_{\rm max}]$ that never intersects $C_{Y}$. Let $\Sigma$ be a bulk Cauchy slice as above containing $C_Y$, and let $p=\gamma\cap \Sigma$. Then $p\notin C_Y$. Because $p\in J^{+}[Y_{\rm max}]\cap J^{-}[Y_{\rm max}]$, there exists a causal curve from $Y_{\rm max}$ to $p$. But then there exists a causal curve from $Y_{\rm max}$ that never crosses $C_Y$, in contradiction with $Y_{\rm max}\subset D[C_Y]$. This proves (1). 

To prove (2), recall that $C_Y=J^{-}[Y]\cap J^{+}[Y]\cap \Sigma$, where $Y\subset Y_{\rm max}$. Since $Y\subset Y_{\rm max}$, $J^{+}[Y]\cap J^{-}[Y]\subset J^{+}[Y_{\rm max}]\cap J^{-}[Y_{\rm max}]$. Thus $C_{Y}\subset J^{+}[Y_{\rm max}]\cap J^{-}[Y_{\rm max}]$. This completes the proof. 

\end{proof}

\noindent \textbf{Proof of Theorem 4.}

\begin{proof}
By Theorem 2, any $C$ satisfying the condition in the theorem will satisfy $D[C]=D[C_{Y}]$ for some $C_{Y}$ defined on some Cauchy slice $\Sigma$, where as usual $C_{Y}\subset J^{+}[Y]\cap J^{-}[Y]$. Since $Y\subset D[{\wt C}]$, by causal convexity of $Y$ and boundary causality, $J^{+}[Y]\cap J^{-}[Y]\subset D[{\wt C}]$. Thus we find $D[C]\subset D[C_{Y}]\subset D[{\wt C}]$.

    To see that $D[C]\cap B=Y$, recall that for any $C$ satisfying Eq.~\ref{eq:Ceqn}, Theorem 1 guarantees that $Y\subset D[C]\cap B$, where $C=J^{+}[Y]\cap J^{-}[Y]\cap \Sigma$ as in Theorem 1. And by the earlier part of the proof, $D[C]\cap B\subset D[\wt C]\cap B=Y$. So $D[C]\cap B=Y$.
\end{proof}

\noindent \textbf{Proof of Theorem 5.}

\begin{proof}
We first establish that $X_{Y}=J^{+}[Y]\cap J^{-}[Y]$. By definition, $X_{Y}\subset J^{+}[Y]\cap J^{-}[Y]$.  Now let $p\in J^{+}[Y]\cap J^{-}[Y]$, and let  $\gamma$ be a causal curve through $p$ with endpoints in $Y$. By AdS topological censorship~\cite{GalSch99} and boundary causality, every such $\gamma$ is fixed-endpoint causally homotopic to a boundary curve, and because $Y$ is causally convex by definition, this boundary curve lies entirely in $Y$. Thus $p\in X_{Y}$ so $J^{+}[Y]\cap J^{-}[Y]\subset X_{Y}$. We thus find that $J^{+}[Y]\cap J^{-}[Y]= X_{Y}$.

We now proceed to show (1): $X_{Y}\subset D[C_{Y}]$. Let $\gamma$ again be a bulk causal curve with endpoints in $Y$. By lemma~\ref{lemma9}, $\gamma$ always intersects $C_{Y}$ (recall that $\sigma\subset C_{Y}$ by definition). Thus for any $p\in J^{+}[Y]\cap J^{-}[Y]$, every inextendible causal curve through $p$ intersects $C_{Y}$; thus $p\in D[C_{Y}]$, so $J^{+}[Y]\cap J^{-}[Y]\subset D[C_{Y}]$; this is simply the statement that the causal wedge is contained within its domain of dependence. 

Claim (2) follows by definition of $C_{Y}$.
\end{proof}

\section{Bulk dual of von Neumann Algebra State Dependence}

In this section we generalize causal wedge reconstruction to a general globally hyperbolic boundary submanifold $Y_{\rm max}$ (i.e. satisfying the max property $Y=D[C_{Y}]\cap B$.

For a boundary spacetime region $Y$, the extrapolate dictionary states that
\be \label{Ex}
 \sM_Y = \sY_Y
\ee
where as before $\sY_Y$ denotes the boundary algebra generated by single-trace operators in $Y$, and $ \sM_U$ denotes the bulk operator algebra in the region $U$.

Below we establish that for $Y=Y_{\rm max}$, ${\cal M}_Y= {\cal M}_{D[C_{Y}]}$, where $C_{Y}$ is the minimal bulk region, as defined in Theorem 3, that satisfies $D[C_{Y}]\cap B=Y$. In this case, the algebra is von Neumann. For $Y\subsetneq Y_{\rm max}$, the boundary subalgebra will be a strict subset of the bulk subalgebra, and the former will not be a von Neumann algebra. This, together with the extrapolate dictionary~\eqref{Ex} establishes our main result: the sector-dependence of boundary von Neumann algebras is reflected in geodesic focusing in the bulk. 

\begin{thm}
    Let $Y\subset B$ be a compact, closed, causally convex region with Cauchy slice $\sigma$ such that $J^{+}[Y]\cap J^{-}[Y]$ does not contain a complete Cauchy slice of $M\cup B$. If $Y=D[C_{Y}]\cap B$, where $C_{Y}=J^{+}[Y]\cap J^{-}[Y]\cap \Sigma$ for some bulk Cauchy slice $\Sigma$ chosen so that $\partial C_{Y}=\partial J^{+}[Y]\cap \partial J^{-}[Y]\cup \sigma$, as guaranteed by Theorem~\ref{Thm:Forwards}, then ${\cal Y}_{Y}$ is a von Neumann algebra and ${\cal Y}_{Y}={\cal M}_{D[C_{Y}]}$. 

    Conversely, if ${\cal Y}_{Y}$ is a von Neumann algebra, then there exists a bulk hypersurface $C$ such that $D[C]\cap B=Y$ and the minimal such $C$ satisfies $\partial C=\partial J^{+}[Y]\cap \partial J^{-}[Y]\cup \sigma$. In this case, ${\cal Y}_{Y}={\cal M}_{D[C_{Y}]}$.
\end{thm}

\begin{proof}

  Consider a region $\tilde Y$ described in Theorem 1 which satisfies $C_{\tilde Y} = C_{Y}$ (Theorem 3). 
By satisfying $D(C_Y) \cap B = Y$, $Y$ stated in the description of the theorem is  the $\tilde Y_{\rm max}$ among all such $\tilde Y$'s. 
    By Lemma 6. $D[C_{Y}]=C_{\tilde Y}''$. Thus ${\cal M}_{D[C_{Y}]}$ is a von Neumann algebra (assuming as we do throughout this paper that reflecting boundary conditions at $\mathscr{I}$ are provided). By Theorem 2, $X_{\tilde Y}$ contains a Cauchy slice of $D[C_{Y}]$. By the bulk time slice axiom, 
    \be\label{y1}
    {\cal M}_{D[C_{Y}]}={\cal M}_{X_{\tilde Y}}\ .
    \ee
    By the bulk timelike tube theorem and Theorem~\ref{thm:Envelope}, 
    \be\label{y2}
    {\cal M}_{X_{\tilde Y}}={\cal M}_{\tilde Y}'' \ .
    \ee
  Suppose there exists a region $\hat Y \supseteq \tilde Y$ such that $\sM_{\hat Y} = {\cal M}_{\tilde Y}''  = {\cal M}_{D[C_{Y}]}$. 
  Then we must have $\hat Y \subseteq D[C_Y] \cap B = Y$, as otherwise we cannot have  $\sM_{\hat Y} = {\cal M}_{D[C_{Y}]}$. 
 Since $Y$ is also a $\tilde Y$, we must have $Y \subseteq \hat Y$. We then conclude that $\hat Y = Y$ and 
 \be 
 {\cal M}_{D[C_{Y}]} = \sM_Y = \sM_Y'' \ .
 \ee
 By the extrapolate dictionary:
    \be
    {\cal M}_{Y}={\cal Y}_{Y} \ .
    \ee
    Thus ${\cal M}_{D[C_{Y}]}={\cal Y}_{Y}$, which is a von Neumann algebra. Thus establishes the forwards direction.

    The converse is as follows: suppose that ${\cal Y}_{Y}$ is a von Neumann algebra. Then ${\cal Y}_{Y}''={\cal Y}_{Y}$. By the extrapolate dictionary, the bulk timelike tube theorem, the bulk time slice axiom, and Theorem~\ref{thm:Envelope}:
    $${\cal Y}_{Y}={\cal M}_{Y}={\cal M}_{Y}''={\cal M}_{X_{Y}}={\cal M}_{D[C_{Y}]},$$
    where $C_{Y}$ is the hypersurface satisfying $\partial C_{Y}=\partial J^{+}[Y]\cap \partial J^{-}[Y]\cup \sigma$ guaranteed by Theorem~\ref{thm:minimalC}.
\end{proof}

We obtain an immediate corollary:
\begin{cor}
    If $Y\neq Y_{\rm max}$, then $\wt {\cal M}_{Y} \subsetneq \wt {\cal M}_{D[C_{Y}]}$ for any $C_{Y}$ satisfying $\partial C_{Y}=\partial J^{+}[Y]\cap \partial J^{-}[Y]\cap \sigma$; in particular, ${\cal Y}_{Y}$ is not a von Neumann algebra.
\end{cor}

\section{Discussion: A Finite-\textit{N} Interpretation }

In this article, we provided a rigorous derivation of the reconstruction of generalized causal wedge algebras for general boundary regions whose single trace algebras at infinite-$N$ are von Neumann algebras. Our derivation identifies the formation of caustics along null congruences as the bulk feature that is dual to the GNS sector-dependence of von Neumann algebras of arbitrary regions at large-$N$. Given a boundary region $Y$, we proved that corresponding region $Y_{\rm max}$, obtained by firing null congruences from $Y$ into the bulk, and then firing null congruences back towards the boundary, supports a von Neumann algebra ${\cal Y}_{Y_{\rm max}}$. This algebra, by the timelike tube theorem and the extrapolate dictionary, is identical to the bulk algebra on $(J^{+}[Y]\cap J^{-}[Y])''$. Furthermore, we conversely proved that the minimal bulk domain of dependence containing $J^{+}[Y_{\rm max}]\cap J^{-}[Y_{\rm max}]$ has the same algebra as $Y_{\rm max}$. 

For technical reasons, our proofs have assumed for convenience that $Y$ is causally convex. We expect that they generalize in a fairy straightforward fashion to causally concave boundary regions. The region $Y_{\rm max}$ is likely still $D[C_{Y}]\cap B$ in this case, with the difference that causally concave regions $Y$ are necessarily proper subsets of $Y_{\rm max}$.

We now proceed to discuss possible \textit{finite}-$N$ extensions of our work. 

\subsection{Speculation on the Duality at finite-\textit{N}}
Consider the boundary theory in a {\it pure}  state $\ket{\Psi}$ dual to a semiclassical bulk spacetime $M$. 
We denote the GNS Hilbert space associated with $\ket{\Psi}$ in the large-$N$ limit as $\sH_\Psi^{(\rm GNS)}$ 
and the finite-$N$ Hilbert space of the boundary theory as $\sH_{\rm CFT}$ (with $N$-dependence suppressed).

Consider the bulk dual region $C_{Y}$ of a boundary globally hyperbolic submanifold $Y$, with 
\be 
\p C_{Y} = \partial J^{+}[Y]\cap \partial J^{-}[Y] \cup \sig \ .
\ee
In the large $N$ limit, for $Y=Y_{\rm max}$ we have the identification 
\be 
\sY_Y =  \sM_{C_{Y}} , \quad \sY_Y' =  \sM_{\overline{C_{Y}}} \ , 
\ee
where $\overline{C_{Y}}$ denotes the complement of $C_{Y}$ on a bulk Cauchy slice $\Sig$. We emphasize that this is the commutant of the algebra rather than the algebra of the causal complement (the prime is applied to ${\cal Y}$ rather than to $Y$). $\sY_Y$ and its commutant $\sY_Y'$ act on $\sH_\Psi^{(\rm GNS)}$. 

What happens at finite $N$? We give evidence below that there exists a novel finite-$N$ extension $\sB_Y$ of ${\cal Y}_{Y}$ which is type $I$ and acts on ${\cal H}_{\rm CFT}$ such that\footnote{A version of the argument below was also used recently in~\cite{LeuLiu25}.} 
\be \label{id1}
\lim_{N \to \infty} S_{\sB_Y} 
= S_{\rm gen} [C_{Y}]  = S_{\rm gen} [\overline{C_{Y}}]  \ .
\ee
$S_\sA$ denotes the entropy associated with an algebra $\sA$. 

In~\eqref{id1},  
\be \label{gen0}
S_{\rm gen} [C_{Y}] = \frac{{\rm Area} (\alpha_{Y}) }{ 4 G_N} +  S_{{\cal M}_{D[C_{Y}]}} , 
\ee
where $ S_{{\cal M}}$ is the bulk entropy in the bulk algebra ${\cal M}$, and we have abbreviated $\partial J^{+}[Y]\cap \partial J^{-}[Y]\equiv \alpha_{Y}$. Equation~\eqref{id1} can be interpreted as giving the leading terms in the $1/N$ expansion of $S_{\sB_Y}$. This is intended to be a statement about convergence: i.e. our proposal is that there is a finite-$N$ algebra $\sB_Y$ that at large-$N$ converges to~\eqref{id1}.
Note that $\sY_Y$ is $\ket{\Psi}$-dependent, and so is $\sB_Y$. 

What kind of ``novel'' extension must $\sB_{Y}$ be? The obvious choice is simply the set of bounded operators localized on $Y$ at finite-$N$. However, it is easy to see that this cannot be the correct extension by considering a time band shape for $Y$. In this case, this choice of extension simply yields ${\cal B}({\cal H}_{\rm CFT})$: the full boundary algebra, whose corresponding entropy is zero. We now argue that another, more natural extension $\sB_Y$ must exist. 



Consider first the bulk algebra $ \sM_{C_{Y}}$  in the $G_N \to 0$ limit. 
It is type III$_1$ in the continuum: an entropy cannot be defined. One way of overcoming this obstacle is via the introduction of a bulk short-distance cutoff $\ep$, which turns the algebra into type I, denoted $S_{ \sM_{C_{Y}}^\ep}$. The specific form of the cutoff is not important (e.g., a lattice regularization).  The entropy $ S_{ \sM_{C_{Y}}^\ep}$ has the form 
\be \label{div1}
S_{ \sM_{C_{Y}}^\ep} = a  \frac{{\rm Area} (\al_Y) }{ \ep^{d-1}} + \cdots, 
\ee
where $a$ is some constant and $\cdots$ denotes less divergent and finite terms in the limit $\ep \to 0$.

From subregion-subalgebra duality, we expect that for $Y=D[C_{Y}]\cap B$
\be \label{sub1}
 \sM_{C_{Y}}^\ep = \sY_Y^\ep, 
\ee
that is, the duality at strictly infinite-$N$ suggests the existence of a regularization on the boundary in the $N \to \infty$ limit that turns $\sY_Y$ into a type I algebra $\sY_Y^\ep$. While  we do not currently know how to describe the regularization explicitly for a general boundary theory, the duality implies that it should exist.  Note that this stage $\ep$ is $G_N$-independent (as we have already taken $G_N \to 0$ limit).

Now consider taking $N$ to be large but finite, or equivalently $G_{N}$ to be finite but small. 
On the gravity side, in the low-energy effective field theory, the coupling $G_N$ is renormalized~(with bare coupling $G_N (\ep)$), and the right hand side of~\eqref{gen0} can be written more precisely as  
\be \label{gen2} 
S_{\rm gen} [C_{Y}] = \lim_{\ep \to 0} \left( \frac{{\rm Area} (\al_Y) }{ 4 G_N (\ep)} +  S_{ \sM_{C_{Y}}^{\ep}} \right) \ .
\ee
It is generally believed that the limit is well defined even though each term on the right hand side cannot be individually defined in the $\ep \to 0$ limit (see e.g.~\cite{SusUgl94, Jac94, LarWil95, Sol11, CooLut13, BouFis15a} for a small selection of works on the subject). 

While equation~\eqref{gen2} can be formally defined in the low-energy effective field theory, 
when $\ep \sim \ell_p$, it is more sensible to interpret the equation in the full quantum gravitational theory. 
In this regime, it is no longer possible to have a clean separation between the $1/G_N$ term in~\eqref{gen2} and the term~\eqref{div1} in $ S_{ \sM_{C_{Y}}^{\ep}}$. Therefore, now it is more sensible to interpret the total contribution $S_{\rm gen} [C_{Y}]$ as the entropy for $ \sM_{C_{Y}}^{\ep}$. In other words, as we decrease $\ep$ all the way to the Planck length $\ell_p$, we expect throughout the process, there exists a type I algebra $ \sM_{C_{Y}}^{\ep}$, and
\be \label{eq:final}
S_{\sM_{C_{Y}^{\ep}} }= S_{\rm gen} [C_{Y}] , \quad \ep \sim \ell_p  \ .
\ee

It is natural to expect the identification~\eqref{sub1} can be similarly extended to this regime, i.e., there exists a finite-$N$ extension of type I algebra $\sY_Y^\ep$, denoted $\sB_Y = \wt \sM_{C_{Y}}^{\ep}$ with $\ep \sim \ell_p$. 
This is our motivation for proposing equation~\eqref{id1}. 

For finite but small $G_N$, $ \sM_{C_{Y}}^{\ep}$ can still be associated with the (approximate) bulk subregion $C_{Y}$. At this level, $C_{Y}$ can only be defined up to quantum gravitational fluctuations, but it is possible that $ \sM_{C_{Y}}^{\ep}$ may still be sharply defined. Despite quantum uncertainties, $D[C_{Y}]$ clearly does not cover the full bulk, and thus  $ \sM_{C_{Y}}^{\ep} = \sB_Y$ cannot be the full algebra. 

We expect that an independent characterization of ${\cal B}_{Y}$ would result in numerous insights from~\eqref{id1}. Here we will highlight one particularly important implication: a nonperturbative understanding of the area term in the generalized second law~\cite{Haw71, Bek72}. Let us briefly remind the reader of the relevant statement: in a classical spacetime satisfying the null energy condition and strong asymptotic predictability, the cross-sectional area of an event horizon generically increases with time~\cite{Haw71}. While successful holographic descriptions of the generalized second law in quantum gravity have been confined to stationary or piecewise stationary horizons, \textit{generic} black hole horizons experience continuous leading order area growth in finite time. Quantitatively, if ${\cal H}^{+}$ is a future event horizon and $\Sigma$ is a Cauchy slice, then the area of ${\cal H}^{+}\cap \Sigma$ is generically strictly smaller than the area of ${\cal H}^{+}\cap \Sigma'$ for any future deformation $\Sigma'$ of $\Sigma$. Under the inclusion of quantum effects, this statement turns into the generalized second law, which has been derived algebraically using algebra inclusions~\cite{ChaPen22a} in the non-generic case in which classical excitations are separated by times longer than the thermal time. The generalized second law has also been hypothesized to have some thermodynamic origin in the dual CFT, although obvious versions of this have been debunked~\cite{EngWal17a}; more fundamentally, the teleological nature of the event horizon likely rules out a naive thermodynamic CFT description.

\begin{figure}
    \centering
    \includegraphics[width=0.5\linewidth]{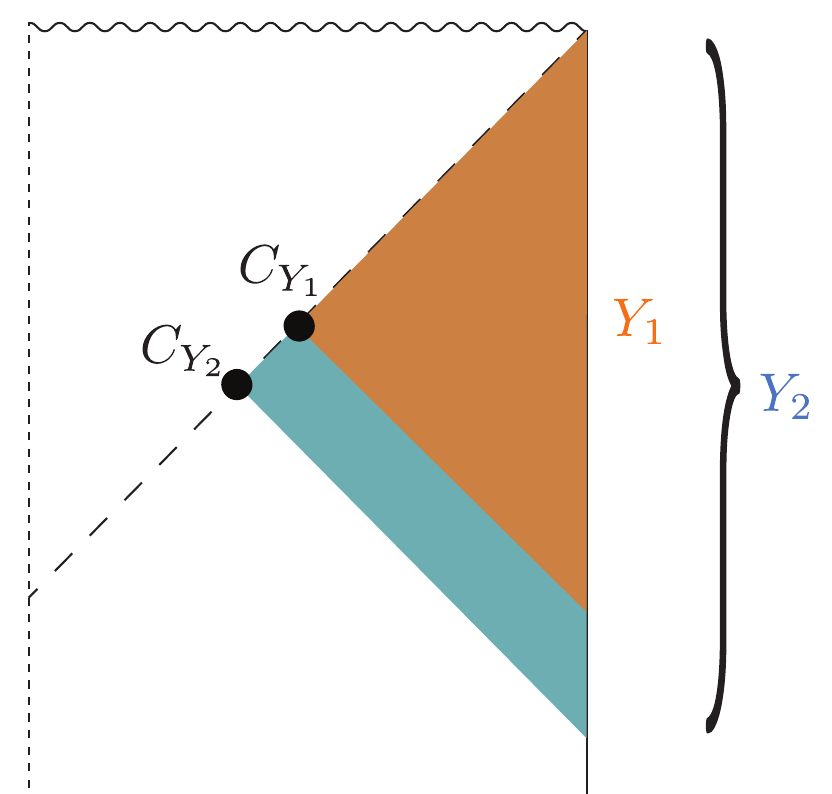}
    \caption{The area law dictates that in generic cases, the area monotonically increase from $C_{Y_{2}}$ to $C_{Y_{1}}$ even if they agree on all but a single generator. The nesting of $Y_{1}$ and $Y_{2}$ captures this whenever $Y_{1}$ are $Y_{2}$ chosen to satisfy the max property $Y=D[C_{Y}]\cap B$.}
    \label{fig:SecondLaw}
\end{figure}

The universality of the area law in gravity under very weak assumptions suggests that the underlying mechanism enforcing it is essential to the emergence of spacetime. Since the area term contributes to $S_{\rm gen}$ at ${\cal O}(G_{N}^{-1})$, this requires a treatment at finite $G_{N}$. To this end, let us apply our tentative proposal above: let $Y_{1}$ and $Y_{2}$ be two semi-infinite nested boundary time bands. Both bands extend infinitely in the future towards $i^{+}$, but end at finite times so that $Y_{1}\subset Y_{2}$. See Fig.~\ref{fig:SecondLaw}. Consider now $C_{Y_{1}}$ and $C_{Y_{2}}$; by construction both $C_{Y_{1}}$ and $C_{Y_{2}}$ will be cuts of the future event horizon, with $C_{Y_{1}}$ to the causal future of $C_{Y_{2}}$. Neither $Y_{1}$ nor $Y_{2}$ are guaranteed to be maximal in the sense of Theorem~\ref{Thm:Forwards}, so we construct the maximal corresponding boundary regions $Y_{1,\rm max}$ and $Y_{2,\rm max}$. Because $C_{Y_{1}}$ and $C_{Y_{2}}$ are causally-separated, the maximal regions should not coincide; in particular, $Y_{1, \rm max}\subset Y_{2,{\rm max}}$. This immediately implies that the corresponding type I algebras nest:
$$\sB_{Y_{1}} \subset \sB_{Y_{2}},$$
and as a result the corresponding entropies are monotonic:
$$S_{\sB_{Y_{1}}}> S_{\sB_{Y_{2}}}.$$
Invoking~\eqref{eq:final} and the proposed equivalence ${\cal B}_{Y}=\widetilde{\sM}^{\epsilon}_{C_{Y}}$, our proposal, if made precise, would yield a CFT derivation of the  generalized second law in complete generality, including spacetimes with dynamically evolving event horizons:
$$S_{\rm gen}[C_{Y_{1}}]>S_{\rm gen}[C_{Y_{2}}].$$

\section*{Acknowledgments}
It is a pleasure to thank L. Ciambelli, E. Gesteau, D. Harlow, V. Hubeny, J. Kudler-Flam, S. Leutheusser, G. Satishchandran, J. Sorce, and M. van Raamsdonk for discussions.  We are also grateful to E. Gesteau, D. Harlow, and J. Sorce for comments on an earlier draft. The work of NE is supported in part by the Department of Energy under Early Career Award DE-SC0021886 and under the HEP-QIS program under grant number DE-SC0025937, by the Heising-Simons Foundation under grant no. 2023-4430,  by the Templeton Foundation via the Black Hole Initiative. HL is supported by the Office of High Energy Physics of U.S. Department of Energy under grant Contract Number  DE-SC0012567 and DE-SC0020360 (MIT contract \# 578218).

\appendix
\section{Technical Lemmas}\label{app:lemmas}
Throughout this appendix we will take $Y\subset B$ be compact, causally convex (as defined in Section 1.1) with a boundary Cauchy-splitting Cauchy slice such that there exists a \textit{bulk} Cauchy slice $\Sigma$ (i.e. a Cauchy slice of $M\cup B$) on which $J^{\pm}[Y^{\mp}]\cap \Sigma\cap M \neq \Sigma \cap M$.

\begin{lem}\label{PastFutBdy} $\forall p\in Y^{+}$, every point $q\in \partial^{+}Y$ is either in $J_{B}^{+}(p)$ or is acausal to $p$. Furthermore, $\forall p \in Y^{+}$, $\exists q\in \partial ^{+}[Y]$ such that $q\in J_{B}^{+}(p)$.
\end{lem}

\begin{proof}
We work fully on $B$ for the entirety of this proof. Every deformation of $\sigma$ that maintains achronality of $\sigma$, keeps $\sigma$ in $Y$, and keeps $D_{B}[\sigma]$ fixed will not change the sets $Y$, and subsequently $\partial^{\pm}Y$. Let $\wt \sigma$ be such a deformation that contains $p$. Then because every causal curve starting in $\wt Y^{+}$ (the new $Y^{+}$ defined by $\wt \sigma$) crosses $wt \sigma$ in the past and no such curve crosses $\wt \sigma$ in the future, points in $\wt Y^{+}$ are either in the future of $p$ or acausal to it. The same is true for for $\partial ^{\pm}\wt Y=\partial^{\pm}Y$. Finally, every future directed curve from $\wt \sigma$ must reach $\partial D^{+}_{Y}[\wt \sigma]=\partial ^{+}[Y]$; this proves the last statement. 
\end{proof}

\begin{lem}\label{NewLem3} ${\rm int}[D_{Y}[\sigma]]=I_{B}^{-}[D_{Y}^{+}[\sigma]]\cap I_{B}^{+}[D_{Y}^{-}[\sigma]]$.
\end{lem}

\begin{proof}
\begin{align*}{\rm int}[Y]=\rm int [D_{Y}[\sigma]]  & =I_{Y}^{-}[D_{Y}^{+}[\sigma]]\cap I_{Y}^{+}[D_{Y}^{-}[\sigma]]\\
& =I_{B}^{-}[D_{Y}^{+}[\sigma]]\cap I_{B}^{+}[D_{Y}^{-}[\sigma]]\cap Y \\ 
& \subseteq I_{B}^{-}[D_{Y}^{+}[\sigma]]\cap I_{B}^{+}[D_{Y}^{-}[\sigma]]
\end{align*}
where the first line follows by definition of $D_{Y}$ and from Wald's lemma 8.3.3 and the second line by definition of $I_{Y}^{\pm}$. We would like to prove that the final line is in fact an equality of sets. Let $p\in I_{B}^{-}[D_{Y}^{+}[\sigma]]\cap I_{B}^{+}[D_{Y}^{-}[\sigma]]=I_{B}^{-}[Y^{+}]\cap I_{B}^{+}[Y^{-}]\subset I_{B}^{-}[Y]\cap I_{B}^{+}[Y]$. We would like to show that $p\in Y$. Since $p\in I_{B}^{-}[Y]\cap I_{B}^{+}[Y]$, there exists a causal curve through $p$ with endpoints in $Y$. By causal convexity, $p\in Y$.
\end{proof}


\begin{lem}\label{NewLem4} $D_{Y}[\sigma]=J_{B}^{-}[D_{Y}^{+}[\sigma]]\cap J_{B}^{+}[D_{Y}^{-}[\sigma]]$    
\end{lem}

\begin{proof}
\begin{align}
    D_{Y}[\sigma]& = \overline{{\rm int}D_{Y}[\sigma]} \\
    & = \overline{I_{B}^{-}[D_{Y}^{+}[\sigma]]\cap I_{B}^{+}[D_{Y}^{-}[\sigma]]}\\
    & \subset \overline{I_{B}^{-}[D_{Y}^{+}[\sigma]]}\cap \overline{I_{B}^{+}[D_{Y}^{-}[\sigma]]}\\
    & =J_{B}^{-}[D_{Y}^{+}[\sigma]]\cap J_{B}^{+}[D_{Y}^{-}[\sigma]]
\end{align}  
where the second line follows by Lemma~\ref{NewLem3} and the fourth line follows by global hyperbolicity of $B$. 
Now let $p\in J_{B}^{-}[D_{Y}^{+}[\sigma]]\cap J_{B}^{+}[D_{Y}^{-}[\sigma]]$. 
$$p\in J_{B}^{-}[D_{Y}^{+}[\sigma]] \ \Rightarrow \ p \in D_{Y}^{+}[\sigma] \ \ {\rm or} \ \ p\in J_{B}^{-}[\sigma]$$
$$p\in J_{B}^{+}[D_{Y}^{-}[\sigma]] \ \Rightarrow \ p \in D_{Y}^{-}[\sigma] \ \ {\rm or} \ \ p\in J_{B}^{+}[\sigma]$$
There are thus four possibilities:
\begin{enumerate}
    \item $p\in D_{Y}^{+}[\sigma]$ and $p\in D_{Y}^{-}[\sigma]$;
    \item $p\in D_{Y}^{+}[\sigma]$ and $p\in J_{B}^{+}[\sigma]$;
    \item $p\in J_{B}^{-}[\sigma]$ and $p\in D_{Y}^{-}[\sigma]$;
    \item $p\in J_{B}^{-}[\sigma]$ and $p\in J_{B}^{+}[\sigma]$.
\end{enumerate}
(1) or (4) can only happen if $p\in \sigma$. (2) happens if $p\in D_{Y}^{+}[\sigma]\cap J_{B}^{+}[\sigma]=D_{Y}^{+}[\sigma]$ and (3) happens if $p\in J_{B}^{-}[\sigma]\cap D_{Y}^{-}[\sigma]=D_{Y}^{-}[\sigma]$. Thus we find $p\in \sigma \cup D_{Y}^{-}[\sigma]\cup D_{Y}^{+}[\sigma]=D_{Y}[\sigma]$.
\end{proof}

\begin{lem} \label{DomainofDepLem}
    $D[S]= \{q\in M: J(q)\subset J[S]\}$ whenever $S$ is closed and achronal. 
\end{lem}

\begin{proof}
    $J(q)\subset J[S] \ \Rightarrow \ q\in J^{+}[S]$ or $q\in J^{-}[S]$ or $q \in S$. If $q\in S$ then $q\in D[S]$ and we are done. Suppose $q\in J^{+}[S]$ or $q\in J^{-}[S]$ but not in $S$, and not in the causal complement $S'$ of $S$. 
    Assume WLOG that $q\in J^{+}[S]$ (the time reverse follows mutatis mutandis). Then $J^{+}(q)\subset J^{+}[S]$ and $J^{-}(q)\subset J^{+}[S]\cup J^{-}[S]$. In particular, consider a past-directed causal curve $\gamma$ from $q$. Since $q\in J^{+}[S]$, $\gamma$ starts out in $J^{+}[S]$. Suppose at some point $\gamma$ exits $J^{+}[S]$. It can do so by crossing the null congruence that generates $\partial J^{+}[S]$ or by crossing $S$. If every such $\gamma$ crosses $S$ then $p\in D^{+}[S]$ and we are done. Suppose it crosses $\partial J^{+}[S]$ instead. Then by achronality of $S$, $\gamma$ must enter $S'$. But that contradicts $J^{-}(q)\subset J^{+}[S]\cup J^{-}[S]$. So every past-directed inextendible causal curve from $q$ must cross $S$ and $q\in D^{+}[S]$. Thus $J(q)\subset J[S]$ and $q\in J^{\pm}[S]$ imply $q\in D^{\pm}[S]$; $J(q)\subset J[S]$ and $q\in S$ also implies $q\in D[S]$. Since that covers all the options, we have one direction: $\{q\in M: J(q)\subset J[S]\}\subset D[S]$. 

    Now for the opposite direction. Let $q\in D^{+}[S]$. Then every past-directed inextendible causal curve from $q$ intersects $S$. This immediately implies $J^{+}(q)\subset J^{+}[S]$. Furthermore let $Q=J^{-}(q)\cap S$. Then $J^{-}(q)\subset J^{+}[Q]\cup J^{-}[Q]$. This implies $J^{-}(q)\subset J^{+}[S]\cup J^{-}[S]=J[S]$. Thus $J(q)\subset J[S]$. The same proof applies mutatis mutandis in the time reverse.
\end{proof}

\begin{lem}\label{CausalCompletionLem}
    For a closed achronal codimension-1 set $S$ which is compactly supported in some Cauchy foliation, $D[S]=S''$. 
\end{lem}

\begin{proof}
Let $r\in S''$. Then by definition, $r\in M\backslash J[S']$: no causal curves through $r$ can reach $S'$. By definition of $S'$ and by global hyperbolicity of $B$ or $M\cup B$ (depending on which manifold $S$ is codimension-1 in), $S'\cup S$ contains a complete Cauchy slice $\Sigma$ of the full spacetime such that $S\subset \Sigma$~\cite{BerSan05, HeaHub14}. If $r\in J^{+}[\Sigma]$, then because $r\in S''$, every past-directed curve from $r$ intersects $\Sigma\backslash S'=S$.  Thus $r\in D^{+}[S]$. Similarly for the time reverse. Thus $S''\subset D[S]$.

Now for the other direction. $p\in D[S]$ means that all inextendible causal curves through $p$ must also go through $S$. That means that none of these causal curves can go through $S'$. Thus $p$ is acausal to $S'$: $p\in S''$.
\end{proof}

\begin{lem}\label{LemmaA} The following four statements hold under the assumptions in this paper:
 \begin{enumerate}  
\item $J^{\pm}_{B}[Y^{\mp}]=J^{\pm}_{B}[Y]$
\item $J^{\pm}[Y^{\mp}]=J^{\pm}[Y]$
\item $J^{\pm}_{B}[Y^{\mp}]=J^{\pm}_{B}[\partial^{\mp}Y]$
\item $J^{\pm}[Y^{\mp}]=J^{\pm}[\partial^{\mp}Y]$
\end{enumerate} 
\end{lem}

\begin{proof}
We will show all of these for one time direction; the opposite time direction follows immediately.

\noindent (1): $Y^{-}\subset Y$ $\Rightarrow$ $J^{+}_{B}[Y^{-}]\subset J^{+}_{B}[Y].$ By Lemma~\ref{NewLem4}, $Y=J^{-}_{B}[Y^{+}]\cap J^{+}_{B}[Y^{-}]$. So:
$$J^{+}_{B}[Y]=J^{+}_{B}[J^{-}_{B}[Y^{+}]\cap J^{+}_{B}[Y^{-}]]\subset J^{+}_{B}[J_{B}^{+}[Y^{-}]]=J_{B}^{+}[Y^{-}].$$
(2) The same reasoning above yields $Y^{-}\subset J^{+}[Y]$ and $J^{+}[Y]\subset J^{+}[J^{+}_{B}[Y^{-}]]$. Because the bulk respects boundary causality $J^{+}[J_{B}[Y^{-}]]=J^{+}[Y^{-}]$, which yields the desired result. 
\noindent (3) One direction of the inclusion follows immediately: since $\partial^{\pm} Y\subset Y^{\pm}$, it is clear that $J^{\pm}[\partial^{\mp}Y]\subset J^{\pm}[Y^{\mp}]$. For the other direction, let $p\in J_{B}^{+}[Y^{-}]$. Then there exists a past-directed boundary causal curve from $p$ to $Y^{-}$. Extend this curve continuously towards the past. Because $\partial^{-}Y\neq \varnothing$ by definition of $Y$, eventually this curve must cross $\partial^{-}Y$. So $p\in J^{+}_{B}[\partial ^{-}Y]$. 

\noindent (4) The first inclusion works identically as in (3). For the other direction, let $p\in J^{+}Y^{-}]$.  Then there exists a past-directed bulk causal curve from $p$ to $Y^{-}$. Since every point in $Y^{-}$ is to the future of some subset of $\partial ^{-}Y$ by Lemma~\ref{PastFutBdy}, $p\in J^{+}[\partial^{-}Y]$.  
\end{proof}

\begin{lem} $\partial J^{+}[Y]\cap \partial J^{-}[Y]\cap B=\partial_{B} \sigma$, where by $\partial _{B}\sigma$ we mean the boundary of $\sigma$ on some Cauchy slice of $B$.  \label{lemma3}
\end{lem}

\begin{proof}
Because the bulk by assumption respects boundary causality, $\partial J^{+}[Y]\cap \partial J^{-}_{B}[Y]=   \partial J^{+}_{B}[Y]\cap  \partial J^{-}_{B}[Y]$. Since $\partial_{B} Y=\partial^{-}Y\cup \partial ^{+}Y\cup \partial_{B} \sigma$, and $\partial_{B} \sigma$ is achronally-separated from both $\partial^{\pm }Y$ (since $Y\subset D[\sigma]$), $\partial_{B} \sigma\subset \partial J^{\pm}_{B}[Y]$. /in particular, $
\partial_{B}\sigma \subset \partial J_{B}^{+}[Y]\cap \partial J_{B}^{-}[Y]=\partial J^{+}[Y]\cap \partial J^{-}[Y]\cap B$. This establishes one direction. 

Now let $p\in \partial J^{+}_{B}[Y]\cap  \partial J^{-}_{B}[Y] = \partial J^{+}_{B}[\partial ^{-}Y]\cap  \partial J^{-}_{B}[\partial ^{+}Y]$, where the equality follows by Lemma~\ref{LemmaA}.  Therefore there exists an achronal future (past) directed null geodesic from $\partial^{-}[Y]$ ($\partial^{+}[Y]$) to $p$. In particular, $p$ is \textit{not} chronally-separated from either $\partial^{\pm}Y$, which are themselves achronal by definition.  

By definition of global hyperbolicity, for all $q\in {\rm int}[Y^{\pm}]$,  all past (future)-directed past (future)-inextendible chronal curves from $q$ intersect $Y^{\mp}$, and therefore eventually also intersect $\partial ^{\mp}Y$.  Therefore $p$ cannot lie in int$[Y]$. If $p\in \partial_{B} Y-\partial_{B} \sigma$, then $p\in \partial ^{+}[Y]$ or $\partial ^{-}Y$ by definition. In either case there is a chronal curve from each $q\in \partial^{\pm}Y$ to $\partial^{\mp}Y$. So $p\notin \partial_{B} Y-\partial_{B} \sigma$. If $p\in Y$, then the only remaining possibility is $p\in \partial_{B} \sigma$ and we are done.  Suppose $p\notin Y$.  Because $p\in \partial J^{+}_{B}[Y^{-}]$,   $p\in D^{+}[\partial\Sigma]$,  where $\partial\Sigma$ is any boundary Cauchy slice containing $\sigma$.  In particular,  $p\in \partial J^{+}_{B}[Y^{-}]\cap D^{+}[\partial\Sigma]$.  Since $Y^{-}\subset D^{-}[\sigma]$, every causal curve from $Y^{-}$ intersects $\sigma$, so $p\in J^{+}[\sigma]$. Similarly, doing the same for $Y^{+}$, we find $p\in J^{-}[\sigma]$. Thus $p \in J^{+}[\sigma]\cap J^{-}[\sigma]$. But $\sigma$ is acausal, so $J^{+}[\sigma]\cap J^{-}[\sigma]=\sigma$.  Since ${\rm int}[\sigma]\subset {\rm int} [Y]$, we find $p\in \partial_{B} \sigma$.
\end{proof}

\begin{lem}\label{Lemma7}
There exists a Cauchy slice $\Sigma$ of $M\cup B$ such that:
\begin{enumerate}
    \item $\sigma \subset \Sigma$
    \item $\partial J^{+}[Y]\cap \partial J^{-}[Y]\subset \Sigma$
\end{enumerate}
\end{lem}

\begin{proof}
Because $Y$ is a timelike hypersurface in $M\cup B$, $\sigma$ and $\partial J^{+}[Y]\cap \partial J^{-}[Y]$ are each acausal and compact in the conformal completion $M\cup B$, so  $\Sigma$ can be picked to fit either one on its own~\cite{BerSan05}; the question is whether they are mutually acausal. By Lemma~\ref{lemma3}, $\partial J^{+}[Y]\cap \partial J^{-}[Y]\cap B=\partial _{B}\sigma$. So there exists a $\Sigma$ containing $\partial J^{+}[Y]\cap \partial J^{-}[Y]$ and $\partial _{B}\sigma$. It remains to show that this $\Sigma$ can also contain int$[\sigma]$. Suppose not. Then $\exists p \in$ int$_{B}[\sigma]$ which is chronal to some $q$ in $\partial J^{+}[Y]\cap \partial J^{-}[Y]$. That is, there exists a chronal curve, which we WLOG take to be future-directed, from $p$ to a point $q\in \partial J^{+}[Y^{-}]\cap \partial J^{-}[Y^{+}]$. Since $p\in {\rm int}_{B}[\sigma]$, $p\in I^{+}[\partial^{-}Y]$, which means that there exists a future-directed chronal curve from $\partial ^{-}Y$ to $p$ and a future-directed causal curve from $p$ to $q$. Thus there exists a chronal future-directed curve from $\partial^{-}Y$ to $q$. But $q\in \partial J^{+}[\partial^{-}[Y]]$, and since $\partial^{-}Y$ is achronal, $\partial^{-}[Y]\subset \partial J^{-}[\partial^{-}Y]$. Since no two points on $\partial J^{-}[\partial^{-}Y]$ can be chronally-separated, we find a contradiction. 
\end{proof}

\begin{lem}\label{lemma9}
    For any choice of $\Sigma$ that contains $\sigma$, for $C\equiv J^{+}[Y]\cap  J^{-}[Y]\cap \Sigma$, every bulk causal curve from $Y^{-}$ to $Y^{+}$ intersects $C$. Furthermore, every bulk \textit{chronal} curve from $Y^{-}$ to $Y^{+}$ intersects int$|_{\Sigma}[C]$.
\end{lem}

\begin{proof}
    By definition, every bulk causal curve $\gamma$ from $Y^{-}$ to $Y^{+}$ lives in $J^{+}[Y^{-}]\cap J^{-}[Y^{+}]$. Since $\sigma$ lies between $Y^{+}$ and $Y^{-}$, every such $\gamma$ has  (in the conformal extension) a future-endpoint in $J^{+}[\sigma]$ and a past-endpoint in $J^{-}[\sigma]$. Because the $\partial^{\pm}Y$ are achronal, every smooth deformation of $\gamma$ to a boundary curve $\gamma'$ with the same endpoints will result in a $\gamma'$ that intersects $\sigma$; note that by AdS topological censorship such deformations always exist~\cite{GalSch99}. Since we are picking $\Sigma$ to contain $\sigma$ and $\Sigma$ is spacetime-splitting by definition of a Cauchy surface, every smooth deformation of $\gamma$ that preserves its chronality and its endpoints must intersect $\Sigma$. Thus every bulk causal curve from $Y^{-}$ to $Y^{+}$ is contained in $J^{+}[Y^{-}]\cap J^{-}[Y^{+}]$ and intersects $\Sigma$. Therefore it must intersect $C$.

    We now turn to the chronal case. For a finite collection of sets, the intersection of the interiors is the interior of the intersection. Thus int$|_{\Sigma}[C]=I^{+}[Y]\cap I^{-}[Y]\cap \Sigma$. Thus the proof above follows for chronal curves to yield the desired result. 
\end{proof}

\begin{lem}\label{LemmaNew}
    $ \partial J^{+}[Y^{-}]\cap \partial J^{-}[Y^{+}]$ is a complete cross-section of the $\partial J^{\pm}[Y^{\mp}]$ congruences.
\end{lem}

\begin{proof}
By causal convexity, every point in $\partial ^{+}Y$ is in the future of $\partial^{-} Y$, but since $\partial ^{-}Y$ is a nonempty and compact past boundary, every point in $\partial^{+}Y$ which is evolved sufficiently far back into the past must exit the future of $Y^{-}$. Thus we find that $\partial J^{+}[Y^{-}]$ divides $J^{-}[Y^{+}]$ into points that are still in the future of $Y^{-}$ (i.e. points in $Y$) and points that are on curves after they have already left the future of $Y^{-}$. By continuity, the same holds for $\partial J^{-}[Y^{+}]$: every point on $\partial J^{-}[Y^{+}]$ either has a past-directed curve towards $Y^{-}$, in which case the point can be deformed along $\partial J^{-}[Y^{+}]$ towards the past until this condition is false, or not. In particular, there is no point on $\partial^{+}Y$ that can be arbitrarily deformed towards the past along the congruence $\partial J^{-}[Y^{+}]$ without crossing $\partial J^{+}[Y^{-}]$.
Thus $\partial J^{+}[Y^{-}]$ is a complete cross-section of $\partial J^{-}[Y^{+}]$ and vice-versa. 
\end{proof}

\begin{lem}\label{lem:New013} There exists a choice of Cauchy slice $\Sigma$ of $M\cup B$  such that $\partial C_{Y}\subset  \partial J^{+}[Y]\cap \partial J^{-}[Y]\cup \sigma$ and edge$[C_{Y}]=\partial_{B}\sigma$. In fact, for every such $\Sigma$, $\partial C_{Y}=  \partial J^{+}[Y]\cap \partial J^{-}[Y]\cup \sigma$. Here as usual $\sigma$ is a Cauchy slice of $Y$.
\end{lem}

\begin{proof}
We first show $C_{Y}\cap B=\sigma$ and thus Edge$[C_{Y}]=\partial_{B}\sigma$, by definition of Edge. We can always pick $\Sigma$ to contain a Cauchy slice $\sigma$ of $Y$.  By boundary causality, since $C_{Y}\cap B=J^{+}[Y]\cap J^{-}[Y]\cap \Sigma\cap B$, $C_{Y}\cap B=J^{+}_{B}[Y]\cap J^{-}_{B}[Y]\cap \Sigma$. By Lemma~\ref{NewLem4} and Lemma~\ref{LemmaA}, $J^{+}_{B}[Y]\cap J^{-}_{B}[Y] = Y$, $C_{Y}\cap B=Y\cap \Sigma$. Since $\sigma$ is a Cauchy slice of $Y$ and $\sigma\subset \Sigma$, we find $C_{Y}\cap B=\sigma$ and Edge$[C_{Y}]=\partial _{B}\sigma$.

We now turn to $\partial C_{Y}\cap M$. By assumption $C_{Y}$ is not a complete Cauchy slice. Let $A=\Sigma\backslash C_{Y}$. Let $\gamma\subset \Sigma\cap M$ be a bulk curve from $C_{Y}$ to $A$ and assume WLOG that $\gamma$ only crosses from $C_{Y}$ to $A$ once in the segment under consideration. To get from $C_{Y}$ to $A$, $\gamma$ must cross $\partial C_{Y}$ at some point in $M$. Since $C_{Y}=J^{+}[Y]\cap J^{-}[Y]\cap \Sigma$, to exit $C_{Y}$ in $M$, $\gamma$ must cross $\partial J^{+}[Y]$ or $\partial J^{-}[Y]$ while remaining on $\Sigma$. By Lemma~\ref{LemmaNew}, $\partial J^{+}[Y]\cap \partial J^{-}[Y]$ is a complete cross-section of the $\partial J^{\pm}[Y]$ congruences. By Lemma~\ref{Lemma7} there exists a Cauchy slice $\wt \Sigma$ such that $\partial J^{+}[Y]\cap \partial J^{-}[Y]\subset \wt \Sigma$ and $\sigma\subset \wt \Sigma$. Let $\Sigma$ be such a $\wt \Sigma$. Then it is impossible to cross $\partial C_{Y}$ in $M$ without crossing $\partial J^{+}[Y]\cap \partial J^{-}[Y]$. Thus there exists a choice of $\Sigma$ such that $\partial C_{Y}\cap M\subset \partial J^{+}[Y]\cap \partial J^{-}[Y]$ and $\sigma\subset \Sigma$. 

We now turn to showing that for such a $\Sigma$, $\partial C_{Y} = (\partial J^{+}[Y]\cap \partial J^{-}[Y])\cup \sigma$. One direction of the inclusion is already established above, so we turn to the other direction. We have already shown that $C_{Y}\cap B=\sigma$, so it remains to treat $C_{Y}\cap M$. Let $p\in \partial J^{+}[Y]\cap \partial J^{-}[Y]$.

By Lemma~\ref{LemmaA}, $\forall p\in \partial J^{+}[Y]\cap \partial J^{-}[Y]$, there exists a piecewise null bulk causal curve $\gamma$ from $\partial^{-}Y$ to $\partial^{+}Y$ through $p$ where each piece from $\partial^{\pm}Y$ to $p$ is achronal. By Lemma~\ref{lemma9}, every bulk causal curve from $Y^{-}$ to $Y^{+}$ intersects $C_{Y}$.  $\gamma$ could (1) intersect the interior of $C_{Y}$ at $p$, in which case $p\in {\rm int} [C_{Y}]$, or it could (2) intersect $C_{Y}$ but not at $p$, or it could (3) intersect $\partial C_{Y}$ at $p$. We would like to rule out options (1) and (2). If $p\in {\rm int}[C_{Y}]$, every small deformation of $p$ on $\Sigma$ would leave it in $C_{Y}$. But $p\in \partial J^{\pm}[Y]$, so there are small deformations of $p$ on $\Sigma$ that take it out of $J^{\pm}[Y]$ (since $\Sigma$ is by definition acausal). Since $C_{Y}=J^{+}[Y]\cap J^{-}[Y]\cap \Sigma$, this would also take $p$ out of $C$. So $p\notin C_{Y}$. This rules out (1). If $\gamma$ intersects $C_{Y}$ but not at $p$, then by Lemma~\ref{lemma9}, $\gamma$ intersects $C_{Y}$ either in the past of $p$ or in the future of $p$. Take WLOG the past. There are no chronal curves from $p$ to $Y^{\pm}$ since $p\in \partial J^{+}[Y]\cap \partial J^{-}[Y]$; if $\gamma$ intersects int$[C_{Y}]$ at some point $q$, then there is a chronal curve from $q$ to $\partial^{-}Y$. So $p$ is chronally-separated from $\partial^{-}Y$ even though it lies on $\partial J^{+}[\partial^{-}Y]$, leading to a contradiction. So $\gamma$ must intersect $\partial C_{Y}$. Since a broken null curve is also chronal, we get the same contradiction unless $p\in\partial C_{Y}$ ($C_{Y}$ cannot itself be causal since we have assumed all Cauchy slices are acausal).
\end{proof}

\begin{lem}\label{Completeness} For any Cauchy slice $\Sigma$ satisfying $\partial C_Y = (\partial J^{+}[Y^{-}]\cap \partial J^{-}[Y^{+}])\cup \sigma$, $J^{+}[Y^{-}]\cap \Sigma=J^{-}[Y^{+}]\cap\Sigma =J^{+}[Y^{-}]\cap J^{-}[Y^{+}]\cap \Sigma=C_{Y}$.    
\end{lem}

\begin{proof}
By contradiction. Assume there exists a Cauchy slice $\Sigma$ containing $\partial C_{Y}$ as above where $J^{+}[Y^{-}]\cap \Sigma\neq C_{Y}$. The time reverse follows \textit{mutatis mutandis}. Let $p\in J^{+}[Y^{-}]\cap \Sigma / C$. By Lemma~\ref{LemmaNew} $\partial C_{Y}$ is a complete cross-section of $\partial J^{+}[Y^{-}]$; thus since $\Sigma$ is acausal, $p$ can only live in $\partial J^{+}[Y^{-}]$ if it is in $\partial C_{Y}\subset C_{Y}$. Since by assumption $p\notin C_{Y}$, we find that $p\in I^{-}[Y^{+}]$. Consider all possible curves that start at $p$ and propagate on $\Sigma$. By assumption, $\exists q\in \Sigma$, $q\notin J^{+}[Y^{-}]$. So there exist curves from $p$ on $\Sigma$ that exit $J^{+}[Y^{-}]$ and thus must cross $\partial J^{-}[Y^{+}]\cap \Sigma$. Since $\partial C_{Y}$ is a complete cross-section of $\partial J^{-}[Y^{+}]$ and $\Sigma$ is acausal, this means that all curves from $p$ must cross $\partial C_{Y}$ in order to exit $J^{-}[Y^{+}]$. But then there is an $O$ which is open in $\Sigma$, where $O\subset J^{-}[Y^{+}]$ and $O$ contains points that live in $\partial J^{-}[Y^{+}]$. This is a contradiction with the definition of $\partial J^{-}[Y^{+}]$. So $p\in C_{Y}$.
\end{proof}

\bibliographystyle{jhep}
\bibliography{all}
\end{document}